\documentclass[11pt]{article}
\usepackage{lmodern}

\usepackage[margin=1in]{geometry}
\usepackage{tikz}
\usepackage{bm}
\usepackage[normalem]{ulem}
\usepackage{parskip}
\usepackage{amsmath}
\usepackage{amssymb}
\usepackage{xspace}
\usepackage[nodate]{datetime} 

\newcounter{theorem} 
\newtheorem{xdefinition}[theorem]{Definition}
\newtheorem{xobservation}[theorem]{Observation}
\newtheorem{xtheorem}[theorem]{Theorem}
\newtheorem{xlemma}[theorem]{Lemma}     
\newtheorem{xproposition}[theorem]{Proposition}
\newtheorem{xcorollary}[theorem]{Corollary}
\newenvironment{definition}{\begin{xdefinition}\rm}%
{\hspace*{\fill}\raisebox{-1pt}{\boldmath$\Box$}\end{xdefinition}}
{\hspace*{\fill}\raisebox{-1pt}{\boldmath$\Box$}\end{xobservation}}
\newenvironment{theorem}{\begin{xtheorem}\rm}{\end{xtheorem}}
\newenvironment{lemma}{\begin{xlemma}\rm}{\end{xlemma}}
\newenvironment{proposition}{\begin{xproposition}\rm}{\end{xproposition}}

\newenvironment{proof}{\begin{trivlist}\item[]{\bf Proof }}%
{\hspace*{\fill}\raisebox{-1pt}{\boldmath$\Box$}\end{trivlist}}

\newcommand{\OPT}{\ensuremath{\operatorname{\textsc{Opt}}}\xspace}
\newcommand{\OFF}{\ensuremath{\operatorname{\textsc{Off}}}\xspace}
\newcommand{\GREEDY}{\ensuremath{\operatorname{\textsc{Greedy}}}\xspace}
\newcommand{\ALG}{\ensuremath{\operatorname{\textsc{Alg}}}\xspace}

\newcommand{\NAT}{\ensuremath{\mathbb{N}}}
\newcommand{\prio}{\ensuremath{\pi}\xspace}
\newcommand{\release}{\ensuremath{r}\xspace}
\newcommand{\length}{\ensuremath{\ell}\xspace}
\newcommand{\compl}{\ensuremath{C}\xspace}
\newcommand{\remaining}{\ensuremath{\ell}\xspace}

\newcommand{\eps}{\varepsilon}
\newcommand{\alg}{\textsc{Alg}\xspace}
\newcommand{\nrpackets}{h}

\begin{document}

\title{Forwarding Packets Greedily on the Line\thanks{The first four authors were supported in part by the Independent Research Fund Denmark, Natural Sciences, grants DFF-0135-00018B and DFF-4283-00079B.}}

\author{
  \begin{tabular}{ll}
        \begin{tabular}{l}
          Joan Boyar \\
          {\small\textsc{University of Southern Denmark}} \\
          \texttt{joan@imada.sdu.dk}
        \end{tabular}
        &
        \begin{tabular}{l}
          Lene M. Favrholdt \\
          {\small\textsc{University of Southern Denmark}} \\
          \texttt{lenem@imada.sdu.dk}
        \end{tabular}
        \\[5ex]
        \begin{tabular}{l}
          Kim S. Larsen \\
          {\small\textsc{University of Southern Denmark}} \\
          \texttt{kslarsen@imada.sdu.dk}
        \end{tabular}
        &
        \begin{tabular}{l}
          Kevin Schewior \\
          {\small\textsc{University of Cologne}} \\
          \texttt{k.schewior@uni-koeln.de}
        \end{tabular} \\[5ex]
        \begin{tabular}{l}
          Rob van Stee \\
          {\small\textsc{University of Siegen}} \\
          \texttt{rob.vanstee@uni-siegen.de}
        \end{tabular}
        \end{tabular}
}

\date{July 29, 2026}

\maketitle

\begin{abstract}
        We consider the problem of forwarding packets arriving online with their destinations
        in a line network.
    In each time step, each router can forward one packet along the edge to its right, and the packet arrives at the next router one time step later. 
    Packets are forwarded until they reach their destination.
    The flow time of a packet is the elapsed time between its release and its arrival at its destination. The goal is to minimize the maximum flow time.

    This problem was introduced by Antoniadis et al.~in 2014, with a focus on line networks. 
    They proposed several natural algorithms. For one, they proved that it is not $O(1)$-competitive; for others, they claimed analogous lower bounds, seemingly leaving no natural candidate for an $O(1)$-competitive algorithm.

    In this paper, we study a natural algorithm not considered in that work. Our algorithm, simply called \GREEDY, selects packets according to their projected flow time under the assumption that they are not delayed any further. We focus on the special case in which each packet needs to be forwarded by one or two routers; this case captures core difficulties.
    We show that \GREEDY achieves a competitive ratio of exactly $2-2^{1-k}$, where $k$ is the number of active routers in the network.	
        
    We also give the first nontrivial general lower bound, which applies even to randomized algorithms: using the same type of instances as in our lower bound for \GREEDY, we show that no algorithm can be $(4/3-\varepsilon)$-competitive for any $\varepsilon>0$.
\end{abstract}

\section{Introduction}
We consider packet scheduling in a network. Routers often have backlogs of packets awaiting transmission, so a policy is needed to decide which packet to forward at each time. Ideally, such a policy would optimize the Quality of Service (QoS) experienced by application-layer clients. This is difficult, however, because routers are generally not informed of which packets belong to which clients. As in Antoniadis et al.~\cite{ABCFMNP14}, we therefore optimize the QoS of the packets themselves by minimizing their flow time.

More specifically, we study online packet forwarding on a line network. Packets are released in an online manner, each with an origin and a destination router. Time is discrete. In each time step, any number of packets may be released at any number of routers, and each router may forward at most one packet to its neighbor on the right. The buffer at each router has unbounded capacity. Routers do not know the states of other routers. The objective is to minimize the maximum flow time, where the flow time of a packet is the time between its release and its arrival at its destination. 

We measure algorithms by their competitive ratios.
Let \OPT be an optimal offline algorithm for packet routing on the line. For any algorithm \ALG, we let $\ALG(I)$ denote the maximum flow time achieved by \ALG on the input sequence $I$.
For any $c \geq 1$, an online algorithm \ALG is said to be $c$-competitive if there exists a constant $b$ such that, for any input sequence $I$, $\ALG(I) \leq c\cdot \OPT(I)+b$.
The (asymptotic) competitive ratio of \ALG is the infimum over all values $c$ for which \ALG is $c$-competitive.

This work is motivated by the results by Antoniadis et al.~\cite{ABCFMNP14}, who defined this appealingly clean but ``surprisingly challenging'' (a predicate we tend to agree with) problem and also focused on the line network. Antoniadis et al.\ showed that a natural algorithm, Earliest Arrival, is not $O(1)$-competitive, and claimed the same for {three} other natural \emph{classes of} algorithms, seemingly leaving no natural algorithm as a candidate for being $O(1)$-competitive. They posed the existence of such an algorithm as an open problem and revealed that there was a modest wager on its resolution. In the first progress on this problem since their paper appeared more than a decade ago, we study the natural \GREEDY algorithm. 
Although we do not settle the wager, our results advance the understanding of this problem and identify \GREEDY as a natural candidate for being $O(1)$-competitive.

Specifically, we develop a good understanding of the case in which all packets have a \emph{length} of at most two, where the length of a packet is the number of routers that need to forward it. As detailed below, core difficulties already arise for such packets. We therefore hope that our work helps pave the way towards solving the general problem.

\paragraph*{Contributions}

Minimizing the maximum flow time on the line seems to entail a fundamental trade-off between the following two proxy objectives:
\begin{itemize}
    \item Give highest priority to packets with the earliest release times. This is the Earliest Arrival algorithm, which is not $O(1)$-competitive~\cite{ABCFMNP14}, essentially because it ignores how far the packets still have to travel.
    \item Prioritize packets by how far they still have to travel (their remaining \emph{length}). This is the Furthest-To-Go algorithm, which Antoniadis et al.\ claim is not $O(1)$-competitive~\cite{ABCFMNP14}. Indeed, a length-$1$ packet can be repeatedly delayed by a continuous stream of length-$2$ packets, even though forwarding it first would yield a maximum flow time of~$3$. The algorithm fails because it ignores when packets were released, namely the first proxy objective.
\end{itemize}

We formalize that there \emph{is} a non-trivial trade-off by providing the first general lower bound on the achievable competitive ratio. In Section~\ref{sec:generallower}, we show that even randomized algorithms cannot be better than $4/3$-competitive. Our construction repeatedly forces the algorithm to choose between the two proxy objectives. Remarkably, the lower bound requires only packets of length $1$ or~$2$.

We propose an algorithm that we term \GREEDY, which resolves the trade-off in a simple way: it prioritizes packets according to the sum of the two proxy objectives, namely how long they have been in the system plus the distance they still have to travel. (We give a different, equivalent formulation later.) In yet other words, this algorithm prioritizes the packets by the flow time they are going to get under the optimistic assumption that they are not delayed any further---therefore the name. So \GREEDY does not act by a ``proxy objective'' but by the \emph{real} objective. Viewed this way, the reader may agree that this is the most natural policy to consider.

Our main result is an analysis of \GREEDY in the case mentioned above----packets of length $1$ or~$2$. Although this may seem restrictive, the fundamental trade-off already arises in this setting. We prove that \GREEDY is $(2-\frac{1}{2^{k-1}})$-competitive, where $k$ is the number of active routers in the line network, and establish a matching lower bound for \GREEDY. These results are presented in Section~\ref{sec:Greedy}, where we begin with the lower bound to build intuition.

We conjecture that the competitive ratio of \GREEDY is constant, possibly even $2$. Indeed, we do not see how longer packets could be used to raise the lower bound on \GREEDY further. If our conjecture is true, we believe that our techniques will be useful in proving it.

\paragraph*{Related work}
Most previous work on routing algorithms under the adversarial model has focused on stability (whether the number of packets in the system remains bounded over time) and throughput maximization. We refer to~\cite{aiello2000adaptive,borodin2001adversarial,chlebus2012adversarial,leighton1994packet,ostrovsky1997universal,DBLP:conf/infocom/Patt-ShamirR19} for representative papers on stability and to~\cite{awerbuch1993throughput,bohm2021throughput,das2025approximation,hyatt2024approximations,kesselman2004buffer,liu2024online,doi:10.1137/1.9781611975482.9} for work on throughput maximization.

In the aforementioned work~\cite{ABCFMNP14}, Antoniadis et al.\ show that Earliest Arrival and Furthest-To-Go are scalable for the maximum-flow-time objective, i.e., $O(1)$-competitive with arbitrarily small speed augmentation. They also show that no $O(1)$-competitive algorithm exists for average flow time.
Havill~\cite{DBLP:conf/icalp/Havill01} considered online packet routing on bidirectional paths and rings
with the objective of minimizing the makespan. They showed that Longest In System and Moving Priority have competitive ratio 2 on linear arrays and also gave a 2-competitive algorithm for bidirectional ring routing. 
Im and Moseley~\cite{DBLP:conf/spaa/ImM15a} considered speed scaling for job scheduling on machines connected by a tree, with the objective of minimizing total flow time, and gave constant-competitive algorithms with constant speed augmentation. Disser et al.~\cite{DBLP:conf/icalp/DisserKL15} give approximation algorithms for bidirectional scheduling on a path to minimize total completion time and show that this problem is NP-hard. 
Erlebach and Jansen~\cite{667220} study the problem of establishing and completing a given set of calls as early as possible for bidirectional and directed calls in various classes of networks. Adamy et al.~\cite{DBLP:journals/algorithmica/AdamyAAE07} studied call control in rings. Given a ring network with edge capacities and a set of paths, the goal is to find a maximum-cardinality subset of the paths such that no edge capacity is violated. They give a polynomial-time algorithm to solve 
the problem optimally. 

The problem we consider can be viewed as a special case of job shop scheduling with unit processing times~\cite{DBLP:journals/mor/BansalKS06}, in which each job must be processed on a contiguous set of machines in left-to-right order.
Soukhal et al.~\cite{SOUKHAL200532} considered flow shop scheduling on two machines.
Wei and Yuan~\cite{WEI201941} considered two-machine flow shop scheduling with equal processing times. See also Chen~\cite{Chen01021995}. 

\section{Preliminaries}

We let $k$ denote the number of active routers in the line network, i.e., routers with a neighbor to their right that they can forward packets to.
Thus, the last router on the line, which can only receive packets, is ignored.
Similarly, if router $i$ is the last router that needs to forward a packet $p$, we call router $i$ the {\em last router} of $p$, even though the {\em destination} of $p$ is router $i+1$.
When a router forwards a packet, we also say that it {\em processes} the packet.

The release time of a packet $p$ is denoted $\release(p)$.
The {\em length} $\length(p)$ of $p$ is the number of routers that need to process $p$, so $\length(p) \in \{1,2\}$.
At any given time $t \ge \release(p)$,
$p$'s remaining length (the number of routers that still need to
process $p$) is denoted $\remaining(p,t)$. 
The {\em completion time}, $\compl(p)$, of a packet $p$ is the first time $t$ such that $\length(p,t)=0$, and its {\em flow time} is $\compl(p)-\release(p)$.

We define the {\em delay} of packet $p$ at time $t$ by
\[
d(p,t)=t-\release(p)-(\length(p)-\remaining(p,t))
\]
and the {\em priority} of packet $p$ at time $t$ by
\[
\prio(p,t)=\length(p)+d(p,t) = t-\release(p) + \remaining(p,t).
\]
The delay of a packet is thus the number of time steps it has spent in the system minus the number of time steps in which it has been processed so far. Consequently, once a packet's priority equals its eventual flow time, it is processed in every subsequent time step until completion.

Without loss of generality, we assume that \OPT is {\em zealous} on each router, meaning that whenever at least one packet is waiting to be processed on a router, \OPT processes a packet.

\section{The Algorithm \GREEDY}\label{sec:Greedy}
Since the priority of a packet at any point in time is a lower bound on its flow time, a natural greedy choice is to forward packets with high priority. This is what the greedy algorithm does.

\GREEDY: In each time step, every router with at least one waiting packet forwards a packet of highest priority. Ties may be broken arbitrarily. 

By providing matching upper and lower bounds, we show that the competitive ratio of \GREEDY is exactly $2-1/2^{k-1}$ for $k$ routers and packets of lengths~1 and~2.

\subsection{Lower bound for \GREEDY}

We begin with the lower bound on \GREEDY's competitive ratio, thus building intuition from concrete examples before proving the matching upper bound.
With only one router, \GREEDY is optimal, since both it and \OPT are zealous.
As a warm-up to the general result (Theorem~\ref{GreedyLower}), we first consider $k=2$ routers, using the same notation as in the proof of Theorem~\ref{GreedyLower}.

\begin{proposition}
    The competitive ratio of \GREEDY is at least $3/2$, even with only $k=2$ routers.
\end{proposition}
\begin{proof}
	Let $h\ge 4$. The input sequence $I$ contains the following packets:
	\begin{itemize}
		\item At time $0$, $h$ packets are released that need to be processed on router $1$; denote this set by $A_1$.
		\item At time $2$, $h$ packets are released that need to be processed on routers $1$ and $2$; denote this set by $B_1$.
		\item At time $h+3$, $2h$ packets are released that need to be processed on router $2$; denote this set by $B_2$.
	\end{itemize}
	Because the packets in $B_1$ are released two time steps later than those in $A_1$, \GREEDY prioritizes the packets in $A_1$ on router $1$.
	On router~$2$, \GREEDY prioritizes the packets in $B_1$ over those in $B_2$.
    Thus, \GREEDY processes $B_1$ on router $1$ during time steps $h$ to $2h-1$ and on router $2$ during time steps $h+1$ to $2h$.
	Consequently, the packets in $B_2$ are processed during time steps $2h+1$ to $4h$, resulting in a maximum flow time of $3h-2$.
	
	An optimal schedule, in contrast, achieves a maximum flow time of $2h$ by giving priority to the packets in $B_1$ on router $1$, thus processing the packets in $A_1$ during time steps $0$ to $1$ and $h+2$ to $2h-1$, the packets in $B_1$ on router $1$ during time steps $2$ to $h+1$ and on router $2$ during time steps $3$ to $h+2$, and the packets in $B_2$ during time steps $h+3$ to~$3h+2$. 
	This way, both the last packet in $A_1$ and the last packet in $B_2$ have flow time~$2h$. 

    Hence, $\GREEDY(I) = \frac{3}{2}\OPT(I)-2$, yielding a lower bound of $\frac{3}{2}$ on \GREEDY's competitive ratio.
\end{proof}

We now extend this construction to any $k\geq 2$ routers.
The construction for $k=4$ is shown in Figure~\ref{fig:lower_bound_example}.

\begin{theorem}
\label{GreedyLower}
    For $k$ routers, the competitive ratio of \GREEDY is at least $2-\frac{1}{2^{k-1}}$.
\end{theorem}
\begin{proof}
We first define a family of input sequences with parameters $k$ and $h\geq 2$. The parameter $h$ can be arbitrarily large, which is necessary to make the result asymptotic.
The sequence $I_k^h$ consists of the blocks of packets $A_1^h, A_2^h,\ldots,A_{k-1}^h$ and $B_1^h, B_2^h,\ldots,B_{k}^h$ defined below. In the remainder of the proof, we omit the superscript $h$ on the blocks $A_i^h$ and $B_i^h$.

All packets have length~2, except those in $A_1$ and $B_{k}$, which have length~1. This choice strengthens the lower bound by improving its dependence on the number of routers $k$.
The release times of the packets in blocks $A_i$ and $B_i$ are determined by the total number of packets in $B_1,B_2,\ldots,B_{i-1}$, based on:
\[r_i = \left(\sum_{j=1}^{i-1}|B_j|\right)+\max\{0,i-2\}.\]

All packets in a block are released at the same time.
Let $r(X)$ denote the release time for the packets in a block $X$.
For $1 \leq i \leq k-1$:
\begin{itemize}
\item Block $A_i$: $2^{k-1-i} \cdot h$ packets released on router $\max\{1,i-1\}$ at time $r(A_i)=r_i$.
\item Block $B_i$: $(2^{k-1} -2^{k-1-i})h$ packets released on router $i$ at time $r(B_i)=r_i+2$.
\end{itemize}
Block $B_{k}$: $2^{k-1} \cdot h$ packets released on router $k$ at time $r(B_k)=r_{k}+2$.

To compare the largest flow times achieved by \GREEDY and \OPT, respectively, we use the following identities, valid for $1 \leq i \leq k-1$:
\begin{align}
|A_i| + |B_i| & = |B_k| = 2^{k-1}\cdot h \label{eq:AplusB} \\
\sum_{j=1}^i |A_j| &= \sum_{j=1}^i 2^{k-1-j}\cdot h = \left(
2^{k-2} +\cdots +2^{k-1-i}\right)h = \left(2^{k-1} - 2^{k-1-i}\right)h = |B_i|\label{eq:sumAi}
\end{align}

For $k=4$ routers, the schedules of \GREEDY and \OPT are depicted in Figure~\ref{fig:lower_bound_example}.
First, we verify that \GREEDY's schedule is as shown.

Each $B$-block is released two time steps after the corresponding $A$-block, and $A_i$ and $B_i$ are released at least $|B_{i-1}|$ time steps after $A_{i-1}$ and $B_{i-1}$, respectively. Thus, the release times and \GREEDY's priorities $\pi$ imply the following relations at every time $t$ and on every router: 
\begin{itemize}
    \item if a packet $a$ from $A_i$ and a packet $b$ from $B_i$ are both waiting to be processed, then $\pi(a,t)>\pi(b,t)$;
    \item if a packet $b$ from $B_i$ and a packet $a'$ from $A_{i+1}$ are both waiting to be processed, then $\pi(b,t)>\pi(a',t)$. 
\end{itemize}
We show below that 
the first $B_i$ packet arrives on router $i+1$ before $B_{i+1}$ is released,
so that \GREEDY processes packets in order $A_1$, $B_1$, $A_2$, $B_2$, \dots, $B_{k-1}$, $B_k$, as depicted in Figure~\ref{fig:lower_bound_example}.
The following calculations determine the completion time and flow time of the last packet in each block for \GREEDY under this processing order.

Before the first packet in $A_i$ is processed, all blocks $A_j$ and $B_j$ with $j<i$ have already been completed. Moreover, because all packets except $A_1$ and $B_k$ have length~2, they require an additional time step to move from one router to the next. \GREEDY finishes the last packet of $A_1$ at time $e_{\GREEDY}(A_1) =2^{k-2}h$, and its
flow time is 
\begin{equation}
\label{eq:fA1}
f_{\GREEDY}(A_1)=2^{k-2}h.
\end{equation}
For $2 \leq i \leq k-1$, \GREEDY finishes the last packet of $A_i$ by time
\begin{align}
    e_{\GREEDY}(A_i) = \sum_{j=1}^i |A_j|  + \sum_{j=1}^{i-1} |B_j|+i-1. \label{eq:eAi}
\end{align}
Subtracting its release time from Equation~\eqref{eq:eAi} and using Equation~\eqref{eq:sumAi} shows that its flow time is
\begin{align}
    f_{\GREEDY}(A_i) = \sum_{j=1}^i |A_j| +1 = \left( 2^{k-1} -2^{k-1-i}\right) h +1.\label{eq:fAi}
\end{align}

For $1 \leq i \leq k-1$, \GREEDY finishes the last packet of $B_i$ by time
\begin{align}
    e_{\GREEDY}(B_i) = \sum_{j=1}^i (|A_j|  + |B_j|)+i, \label{eq:eBi}
\end{align}
because $B_i$ starts on the router where $A_i$ ends.
Subtracting its release time from Equation~\eqref{eq:eBi} shows that its flow time is
\begin{align}
  &  f_{\GREEDY}(B_1) = 2^{k-1}\cdot h -1 \text{ and} \nonumber \\ 
  &  f_{\GREEDY}(B_i) = \left(\sum_{j=1}^i |A_j|\right) +|B_i| =|B_i|+|B_i| = \left( 2^{k} -2^{k-i}\right) h, \label{eq:fBi}
\end{align}
using Equation~\eqref{eq:sumAi}.

\GREEDY finishes the last packet of $B_k$ by time
\begin{align}
    e_{\GREEDY}(B_k) = \sum_{j=1}^{k-1} (|A_j|  + |B_j|)+|B_k|+k-1, \label{eq:eBk}
\end{align}
because $B_k$ is processed on only one router.
Subtracting its release time from Equation~\eqref{eq:eBk} shows that its flow time is
\begin{align}
    f_{\GREEDY}(B_k) = \left(\sum_{j=1}^{k-1} |A_j|\right) + |B_k| - 1 = |B_{k-1}|+|B_k|-1 = \left( 2^{k} -1\right) h -1,
     \label{eq:fBk}
\end{align}
by Equation~\eqref{eq:sumAi}.

We now verify that $B_{i+1}$ is released only after \GREEDY has completed $A_{i}$ (and hence also $B_{i-1}$), ensuring that \GREEDY is processing $B_i$ on router $i+1$ when $B_{i+1}$ is released. For $1 \leq i \leq k-1$,
\begin{align*}
    r(B_{i+1}) & = \left( \sum_{j=1}^{i} |B_j|\right) + i-1 +2
    = |B_i|+\left( \sum_{j=1}^{i-1} |B_j|\right) + i+1 \\
    & =
     \left( \sum_{j=1}^{i}|A_i|+\sum_{j=1}^{i-1} |B_j|\right)+i+1,~\text{by Equation~\eqref{eq:sumAi}} \\
    & = e_{\GREEDY}(A_{i})+2,~\text{by Equation~\eqref{eq:eAi}.}
\end{align*}

Finally, since $r(A_{i+1})=r(B_{i+1})-2$, the calculation above also gives $e_{\GREEDY}(A_i) = r(A_{i+1})$. Thus,
for every $i\geq 1$, \GREEDY finishes processing $A_i$ exactly when $A_{i+1}$ is released. 

We have now verified all assumptions used above, so Equations~\eqref{eq:eAi}--\eqref{eq:fBk} are valid, as is
the depiction of \GREEDY's schedule in Figure~\ref{fig:lower_bound_example}, with the blocks scheduled one after another in order of their priorities.  Comparing the flow times for $A_i$, $B_i$, and $B_k$ from Equations~\eqref{eq:fA1}, \eqref{eq:fAi}, \eqref{eq:fBi}, and~\eqref{eq:fBk}, we obtain
\begin{align}
    \GREEDY(I_k^h) & = f_{\GREEDY}(B_k)
     =  \left( 2^k-1\right)h-1. \label{flowG}
\end{align}

We now analyze \OPT's schedule; again the reader can use Figure~\ref{fig:lower_bound_example} as a guide.
The depicted schedule for \OPT is feasible: the first packet of each $B$-block starts one time step after its release, and the $A$-blocks have much larger delays relative to their release times.
The adversary releases the $B$-packets earlier than \OPT requires in order to increase the flow time of \GREEDY: \GREEDY's maximum flow time occurs on $B_k$, whereas \OPT's occurs on $A_{k-1}$.

For $1 \leq i \leq k-1$, \OPT finishes the last packet of $B_i$ by time
\begin{align}
    e_{\OPT}(B_i) = \sum_{j=1}^i |B_j|+i+2. \label{eq:OeBi}
\end{align}
Subtracting its release time from Equation~\eqref{eq:OeBi} gives that its flow time is
\begin{align}
  &  f_{\OPT}(B_1) = 2^{k-2} h +1\text{ and} \nonumber \\ 
  &  f_{\OPT}(B_i) = |B_i|+2 = \left( 2^{k-1} -2^{k-1-i}\right) h +2. \label{eq:OfBi}
\end{align}

Block $B_k$ starts when $B_{k-1}$ finishes, at time $\sum_{j=1}^{k-1} |B_j|+(k-1)+2 = r(B_k)+1$, by Equation~\eqref{eq:OeBi}. It finishes after $|B_k|$ further time steps, so its flow time is 
\begin{align}
   f_{\OPT}(B_k) = |B_k|+1 = 2^{k-1}h+1.  \label{eq:OfBk} 
\end{align}

We next consider the blocks $A_i$.
\OPT finishes the last packet of $A_1$ at time $|A_1|+|B_1|=2^{k-1}h$.
Thus, since $r(A_1)=0$, $e_{\OPT}(A_1) = f_{\OPT}(A_1)=2^{k-1}h$.
For $2 \leq i \leq k-1$, \OPT finishes the last packet of $A_i$ by time
\begin{align}
    e_{\OPT}(A_i) & = |A_i| + e_{\OPT}(B_i) -1, \text{ since $A_i$ ends on the router where $B_i$ starts} \nonumber \\
    & = |A_i| + \left(\sum_{j=1}^i |B_j|+i+2\right)-1,~\text{by Equation~\eqref{eq:OeBi}.} \label{eq:OeAi} 
\end{align}
Subtracting its release time from Equation~\eqref{eq:OeAi} gives that its flow time is
\begin{align}
f_{\OPT}(A_i) & = |A_i| + \sum_{j=1}^{i} |B_j|+i+1 -r(A_i) \nonumber \\
    & = |A_i| + |B_i| +3 \nonumber \\
    & = 2^{k-1}\cdot h +3,
    \text{ by Equation~\eqref{eq:AplusB}.} \label{eq:OfAkminus1} 
\end{align}
This flow time is obtained for every $i\geq 2$, and in particular for $A_{k-1}$.

Comparing Equations~\eqref{eq:OfBi}, \eqref{eq:OfBk}, and~\eqref{eq:OfAkminus1}, we see that \OPT's flow time for $A_{k-1}$ is slightly larger than its flow time for $B_k$. Therefore,
\begin{align}
    \OPT(I_k^h) & = f_{\OPT}(A_{k-1})
     =  2^{k-1}\cdot h + 3. \label{flowO}
\end{align}

Solving Equation~\eqref{flowO} for $h$ gives
\begin{align*}
    h = \frac{\OPT(I_k^h)-3}{2^{k-1}}.
\end{align*}

Then,
\begin{align*}
    \GREEDY(I_k^h) & = \left( 2^k -1\right) h-1,~\text{by Equation~\eqref{flowG}}  \\
    & = \left( 2^k -1\right) \frac{\OPT(I_k^h)-3}{2^{k-1}}-1  \\
    & = \left( 2-\frac{1}{2^{k-1}}\right)\OPT(I_k^h) -3\left(2-\frac{1}{2^{k-1}}\right)-1.
\end{align*}

Since $\OPT(I_k^h)$ grows without bound as $h$ increases, the competitive ratio of \GREEDY cannot be smaller than $2-\frac{1}{2^{k-1}}$.
\end{proof}

Note that in the proof of Theorem~\ref{GreedyLower}, the cost of \OPT is
never greater than that of \GREEDY on any prefix of $I_k^h$, so \OPT's schedule is an online-bounded optimal~\cite{BEFLL18} solution for \GREEDY. Thus, \GREEDY has an online-bounded ratio of at least 
$2-\frac{1}{2^{k-1}}$. In online-bounded analysis, the offline optimal algorithm that an online algorithm competes against is constrained to never perform worse than the online algorithm on any prefix, and therefore cannot fully exploit its knowledge of when the input ends.

\begin{figure}[tpb]
\centering
\begin{tikzpicture}[xscale=1.25,yscale=1.00]
\draw[fill=blue!13] (0, 0) -- (1, 0) -- (1, -2.0) -- (0, -2.0) -- (0, 0);
\node[anchor = north west] at (0, 0) {\footnotesize $0$};
\node[anchor = south east] at (1, -2.0) {\footnotesize $4h$};
\node at (0.5, -1.0) {\footnotesize $\bm{4h}$};
\draw[fill=red!13] (0, -2.0) -- (1, -2.0) -- (1, -2.1) -- (2, -2.1) -- (2, -4.1) -- (1, -4.1) -- (1, -4.0) -- (0, -4.0) -- (0, -2.0);
\node[anchor = north west] at (0, -2.0) {\footnotesize $2$};
\node[anchor = south east] at (2, -4.1) {\footnotesize $8h+1$};
\node at (1, -3.0) {\footnotesize $\bm{8h-1}$};
\draw[fill=blue!25] (0, -4.0) -- (1, -4.0) -- (1, -4.1) -- (2, -4.1) -- (2, -5.1) -- (1, -5.1) -- (1, -5.0) -- (0, -5.0) -- (0, -4.0);
\node[anchor = north west] at (0, -4.0) {\footnotesize $4h$};
\node[anchor = south east] at (2, -5.1) {\footnotesize $10h+1$};
\node at (1, -4.5) {\footnotesize $\bm{6h+1}$};
\draw[fill=red!25] (1, -5.1) -- (2, -5.1) -- (2, -5.199999999999999) -- (3, -5.199999999999999) -- (3, -8.2) -- (2, -8.2) -- (2, -8.1) -- (1, -8.1) -- (1, -5.1);
\node[anchor = north west] at (1, -5.1) {\footnotesize $4h+2$};
\node[anchor = south east] at (3, -8.2) {\footnotesize $16h+2$};
\node at (2, -6.6) {\footnotesize $\bm{12h}$};
\draw[fill=blue!37] (1, -8.1) -- (2, -8.1) -- (2, -8.2) -- (3, -8.2) -- (3, -8.7) -- (2, -8.7) -- (2, -8.6) -- (1, -8.6) -- (1, -8.1);
\node[anchor = north west] at (1, -8.1) {\footnotesize $10h+1$};
\node[anchor = south east] at (3, -8.7) {\footnotesize $17h+2$};
\node at (0.5, -8.35) {\footnotesize $\bm{7h+1}$};
\draw[fill=red!37] (2, -8.7) -- (3, -8.7) -- (3, -8.799999999999999) -- (4, -8.799999999999999) -- (4, -12.299999999999999) -- (3, -12.299999999999999) -- (3, -12.2) -- (2, -12.2) -- (2, -8.7);
\node[anchor = north west] at (2, -8.7) {\footnotesize $10h+3$};
\node[anchor = south east] at (4, -12.299999999999999) {\footnotesize $24h+3$};
\node at (3, -10.45) {\footnotesize $\bm{14h}$};
\draw[fill=red!49] (3, -12.3) -- (4, -12.3) -- (4, -16.2) -- (3, -16.2) -- (3, -12.3);
\node[anchor = north west] at (3, -12.3) {\footnotesize $17h+4$};
\node[anchor = south east] at (4, -16.2) {\footnotesize $32h+3$};
\node at (3.5, -14.25) {\footnotesize $\bm{15h-1}$};
\draw[fill=blue!13] (3.5, 0) -- (4.5, 0) -- (4.5, -0.2) -- (3.5, -0.2) -- (3.5, 0);
\node[anchor = north west] at (3.5, 0) {\footnotesize $$};
\node[anchor = south east] at (4.5, -0.2) {\footnotesize $$};
\node at (4.0, -0.1) {\footnotesize $\bm{}$};
\draw[fill=red!13] (3.5, -0.2) -- (4.5, -0.2) -- (4.5, -0.30000000000000004) -- (5.5, -0.30000000000000004) -- (5.5, -2.3000000000000003) -- (4.5, -2.3000000000000003) -- (4.5, -2.2) -- (3.5, -2.2) -- (3.5, -0.2);
\node[anchor = north west] at (3.5, -0.2) {\footnotesize $2$};
\node[anchor = south east] at (5.5, -2.3000000000000003) {\footnotesize $4h+3$};
\node at (4.5, -1.2) {\footnotesize $\bm{4h+1}$};
\draw[fill=blue!13] (3.5, -2.2) -- (4.5, -2.2) -- (4.5, -4.0) -- (3.5, -4.0) -- (3.5, -2.2);
\node[anchor = north west] at (3.5, -2.2) {\footnotesize $0$};
\node[anchor = south east] at (4.5, -4.0) {\footnotesize $8h$};
\node at (4.0, -3.1) {\footnotesize $\bm{8h}$};
\draw[fill=red!25] (4.5, -2.3) -- (5.5, -2.3) -- (5.5, -2.4) -- (6.5, -2.4) -- (6.5, -5.3999999999999995) -- (5.5, -5.3999999999999995) -- (5.5, -5.3) -- (4.5, -5.3) -- (4.5, -2.3);
\node[anchor = north west] at (4.5, -2.3) {\footnotesize $4h+2$};
\node[anchor = south east] at (6.5, -5.3999999999999995) {\footnotesize $10h+4$};
\node at (5.5, -3.8) {\footnotesize $\bm{6h+2}$};
\draw[fill=blue!25] (3.5, -4.0) -- (4.5, -4.0) -- (4.5, -5.0) -- (3.5, -5.0) -- (3.5, -4.0);
\draw[fill=blue!25] (4.5, -5.3) -- (5.5, -5.3) -- (5.5, -6.3) -- (4.5, -6.3);
\node[anchor = north west] at (3.5, -4.0) {\footnotesize $4h$};
\node[anchor = south east] at (5.5, -6.3) {\footnotesize $12h+3$};
\node at (5.0, -5.7) {\footnotesize $\bm{8h+3}$};
\draw[fill=red!37] (5.5, -5.4) -- (6.5, -5.4) -- (6.5, -5.5) -- (7.5, -5.5) -- (7.5, -9.0) -- (6.5, -9.0) -- (6.5, -8.9) -- (5.5, -8.9) -- (5.5, -5.4);
\node[anchor = north west] at (5.5, -5.4) {\footnotesize $10h+3$};
\node[anchor = south east] at (7.5, -9.0) {\footnotesize $17h+5$};
\node at (6.5, -7.15) {\footnotesize $\bm{7h+2}$};
\draw[fill=blue!37] (4.5, -6.3) -- (5.5, -6.3) -- (5.5, -6.8) -- (4.5, -6.8) -- (4.5, -6.3);
\draw[fill=blue!37] (5.5, -8.9) -- (6.5, -8.9) -- (6.5, -9.4) -- (5.5, -9.4);
\node[anchor = north west] at (4.5, -6.3) {\footnotesize $10h+1$};
\node[anchor = south east] at (6.5, -9.4) {\footnotesize $18h+4$};
\node at (5.0, -9.15) {\footnotesize $\bm{8h+3}$};
\draw[fill=red!49] (6.5, -9.0) -- (7.5, -9.0) -- (7.5, -12.9) -- (6.5, -12.9) -- (6.5, -9.0);
\node[anchor = north west] at (6.5, -9.0) {\footnotesize $17h+4$};
\node[anchor = south east] at (7.5, -12.9) {\footnotesize $25h+5$};
\node at (7.0, -10.95) {\footnotesize $\bm{8h+1}$};
\draw[color=gray] (-0.25, -0.0) -- (5.0, -0.0);
\node[anchor = east, fill = blue!13] at (-0.25, -0.0) {\footnotesize \begin{tabular}{@{}c@{}}$A_1$\\$0$\end{tabular}};
\draw[color=gray] (-0.25, -2.0) -- (6.0, -2.0);
\node[anchor = east, fill = blue!25] at (-0.25, -2.0) {\footnotesize \begin{tabular}{@{}c@{}}$A_2$\\$4h$\end{tabular}};
\draw[color=gray] (-0.25, -5.1) -- (7.0, -5.1);
\node[anchor = east, fill = blue!37] at (-0.25, -5.1) {\footnotesize \begin{tabular}{@{}c@{}}$A_3$\\$10h+1$\end{tabular}};
\draw[color=gray] (0, -0.2) -- (7.75, -0.2);
\node[anchor = west, fill = red!13] at (7.75, -0.2) {\footnotesize \begin{tabular}{@{}c@{}}$B_1$\\$2$\end{tabular}};
\draw[color=gray] (0, -2.2) -- (7.75, -2.2);
\node[anchor = west, fill = red!25] at (7.75, -2.2) {\footnotesize \begin{tabular}{@{}c@{}}$B_2$\\$4h+2$\end{tabular}};
\draw[color=gray] (0.5, -5.3) -- (7.75, -5.3);
\node[anchor = west, fill = red!37] at (7.75, -5.3) {\footnotesize \begin{tabular}{@{}c@{}}$B_3$\\$10h+3$\end{tabular}};
\draw[color=gray] (1.5, -8.9) -- (7.75, -8.9);
\node[anchor = west, fill = red!49] at (7.75, -8.9) {\footnotesize \begin{tabular}{@{}c@{}}$B_4$\\$17h+4$\end{tabular}};
\end{tikzpicture}
\caption{Example of the lower-bound construction for $k=4$ routers. The small boxes on the left and right give the release times of the blocks; the choice of side is only for readability and has no algorithmic significance. The horizontal guide lines indicate the corresponding release times. In the center, \GREEDY's schedule is shown on the left and \OPT's on the right. A block that is open to the left indicates that its processing started on the previous router. Within each block, the release time appears at the upper left, the completion time at the lower right, and the flow time in the center.}
\label{fig:lower_bound_example}
\end{figure}

\subsection{Upper bound for \GREEDY}
We now turn to the main result of the paper, the upper bound for \GREEDY, beginning with some definitions.
A packet $p$ is \emph{alive} for a given algorithm at any time $t$ such that $t \geq \release(p)$ and $\length(p,t)>0$.

Fix an optimal offline algorithm, \OPT.
\begin{definition} 
\label{def:a_and_g}
	Let $a_i(t)$ be the number of alive packets in \OPT's schedule that still need to be processed on router $i$ at time $t$. Similarly, let $g_i(t)$ be the number of alive packets that \GREEDY still needs to process on router $i$ at time $t$.
\end{definition}
The optimal flow time is bounded below by the maximum value of $a_i(t)$ over all time steps and routers. 

\begin{definition}
\label{def:delta}
    We define $\Delta_i(t)=g_i(t)-a_i(t)$ for all routers $i$ and times $t$.
\end{definition}
Thus, $\Delta_i(t)$ is the number of packets by which \OPT is ahead on router $i$ at time $t$: it is the number of packets that \OPT has processed on router $i$ before time $t$ minus the number that \GREEDY has processed there before time $t$. 
The value $\Delta_i(t)$ may be positive or negative.

\begin{lemma}
    \label{lem:1}
    For each time $t\geq 0$, we have $\Delta_1(t)=0$.
\end{lemma}
\begin{proof}
    Both \OPT and \GREEDY are zealous, and packets arrive at router~1 at the same times for both algorithms. Therefore, the number of alive packets on router~1 is always the same, and $\Delta_1(t)=0$. 
\end{proof}

We wish to bound $\Delta_i(t)$, the number of packets by which \OPT is ahead on router $i$ at time $t$.
This is the purpose of Lemmas~\ref{lem:di+}--\ref{lemma:DeltaUpper}. 
The basic idea of the proof is the following:
Each unit by which $\Delta_i$ exceeds $\Delta_{i-1}$ corresponds to a packet whose last router is $i-1$ and which \GREEDY has processed there but \OPT has not. The reason \OPT did not process it is that it was processing another packet released on router $i-1$ that contributes to its lead on router $i$. If there are $\ell$ such packets, then \OPT still has to process $\ell$ packets on router $i-1$. Each of these packets must have priority at least that of the packets \OPT processed, because \GREEDY preferred them. Processing them takes \OPT $\ell$ additional steps, giving a lower bound on its cost.

    \begin{lemma}
		\label{lem:di+}
		If $\Delta_i(t+1)=\Delta_i(t)+1 
		> \Delta_{i-1}(t)\ge0$, then, at time $t$, there are $\Delta_i(t+1)-\Delta_{i-1}(t)$ packets with current priority at least $\Delta_i(t+1) + 1$ that \OPT still needs to process on router $i-1$.
	\end{lemma}
	\begin{proof}
		The equality $\Delta_i(t+1)=\Delta_i(t)+1$ implies that \GREEDY is idle on router $i$ at time $t$, so \GREEDY has processed all packets released on router $i$. 
        Thus, by time $t$, \OPT has processed at least $\Delta_i(t+1)$ length-2 packets that were released on router $i-1$ and that \GREEDY has not yet processed there. 
        Let $P$ be this set of packets, and let $p$ be a packet of minimum release time in $P$. Note that 
        \[t-r(p) \geq |P| \geq \Delta_i(t+1),\]
        since \OPT processes $P$ on router $i-1$ between $r(p)$ and $t$.
        
        Now consider the sets of packets that \GREEDY and \OPT process on router $i-1$ before time $t$.
		Call these sets $G$ and $O$, respectively, and note that
        \begin{align}\label{eq:sizeO}
            |O|=|G|+\Delta_{i-1}(t).
        \end{align}
        Moreover, since $P \subseteq O \setminus G$, 
        \begin{align}\label{eq:sizeO-G}
        |O\setminus G|\ge|P|\ge \Delta_{i}(t+1),
        \end{align}
        so 
        \begin{align*}
        |G \setminus O| 
        &= |O\setminus G| - \Delta_{i-1}(t), \text{ by \eqref{eq:sizeO} and removing $O\cap G$ from both $O$ and $G$}\\
        &\geq \Delta_i(t+1)-\Delta_{i-1}(t), \text{ by \eqref{eq:sizeO-G}}.
        \end{align*}
Since all packets in $P$ are released on router $i-1$, they become available to \GREEDY there immediately upon release. 
Therefore, by the definition of \GREEDY and the fact that $p$ has length $2$, each packet in $G \setminus O$ is at least as old as $p$ and of length 2 or is at least one time step older than $p$.
It follows that, at time $t$, \OPT has at least $\Delta_i(t+1)-\Delta_{i-1}(t)$ packets, each with priority at least $\Delta_i(t+1)+1$, that still need to be processed on router $i-1$.
	\end{proof}

    For an illustration of Lemma~\ref{lem:di+}, see the lower-bound construction depicted
in Figure~\ref{fig:lower_bound_example}. This input sequence uses the blocks defined at the beginning of the proof of Theorem~\ref{GreedyLower}. Letting $i=4$, we consider \[t=17h+2,\] since
$17h+2$ is the last time before completion at which \GREEDY is idle on
router~4. 
Since \OPT processes the last two packets of $B_3$ at times $t=17h+2$ and $t+1=17h+3$ on router~$3$, the set $P$ consists of the first $7h-2$ packets in $B_3$, and 
\[\Delta_4(t+1)=\Delta_4(17h+3)= 7h-2.\]
The oldest
packet $p$ in $P$ was released at time \[r(p)=10h+3\] (all packets in $P$ have
this release time). By this time, \OPT has completed all but the last packet of $B_2$, while \GREEDY has completed only the first packet of $B_2$, so \OPT is ahead of
\GREEDY by \[ \Delta_3(r(p))=6h-2\] packets from $B_2$ on router~3.

At time $r(p)$ on router~3, \OPT has not yet begun
processing any of the $h$ packets in $A_3$ on router~3. At time $t=17h+2$, each has
priority \[ \pi(A_3,t)= 2+(17h+2)-(10h+1)-(2-1)=7h+2. \] These are the packets in $G\setminus O$ from Lemma~\ref{lem:di+}. According to the lemma, the
last packet in $A_3$ that \OPT processes on router~3 has priority
(and flow time) at least $\Delta_4(t+1)=7h-2$, and there are at least
\[ |A_3| = \Delta_4(t+1)-\Delta_3(t) = (7h-2)-(6h-2) = h \] packets in $A_3$. In this example, the number of packets remaining for \OPT on router~3 matches the lower bound in the lemma, and the priority $7h+2$ exceeds the guaranteed value $\Delta_4(t+1)+1=7h-1$. As noted
below in the proof of Lemma~\ref{lemma:DeltaUpper}, at least one of these packets has flow time at least
\[ 2\Delta_4(t+1) - \Delta_3(t)= 2(7h-2)-(6h-2)=8h-2 \] in \OPT's schedule
(its actual flow time is $8h+3$).

	\begin{lemma}\label{lemma:DeltaUpper}
		For any router $i \ge 1$, \[\max_{t \ge 0, j \le i} \Delta_j(t) \le \left(1-\frac{1}{2^{i-1}}\right)\OPT.\]
	\end{lemma}
	\begin{proof}
		We prove the lemma by induction on $i$.
		For $i=1$, the claim follows from Lemma~\ref{lem:1}, establishing the base case.
		
		For the inductive step, fix a router $j \geq 2$ and let $t+1$ be the first time at which the maximum value of $\Delta_j$ is attained. 
		By Lemma~\ref{lem:di+}, at time $t$ there are $\Delta_j(t+1)-\Delta_{j-1}(t)>0$ packets that \OPT still needs to process on router $j-1$ and whose priority is already at least $\Delta_j(t+1)+1$.
        The first of these packets may have flow time $\Delta_j(t+1)+1$, but the last has flow time at least $2\Delta_j(t+1)-\Delta_{j-1}(t)$.
		Thus, 
		\begin{align*}
			\OPT & \ge 2\Delta_j(t+1) - \Delta_{j-1}(t) \\
			& \ge 2\Delta_j(t+1) - (1-2^{2-j})\OPT, \text{ by the induction hypothesis}\\
			& \ge 2\Delta_j(t+1) - (1-2^{2-i})\OPT, \text{ since } j \le i.
		\end{align*}
        Hence,
            \[(1-2^{1-i})\OPT \ge \Delta_j(t+1).\]
        This completes the induction.
	\end{proof}

\begin{theorem}
  If there are $k$ routers and all packets have length~1 or~2,
  \GREEDY is $\left(2-\frac{1}{2^{k-1}} \right)$-competitive.
\end{theorem}
\begin{proof}
Fix any router $i\ge2$ and any packet $p$ that is released at time $r(p)$ on router $i-1$ or $i$ and has router $i$ as its last router.
The priority of every packet that remains alive in \GREEDY's schedule increases by 1 in each time step in which it is not processed. 
Therefore, no packet released at time $t'=r(p)+2$ or later, on any router, can delay $p$, because $p$ has higher priority. 

At time $t'$, there are $g_{i-1}(t')$
packets alive (on router $i-2$ or $i-1$) that will eventually need to be processed on router $i-1$, and
\begin{align*} 
g_{i-1}(t') 
&= a_{i-1}(t')+\Delta_{i-1}(t') \\ 
&\le \OPT +\left(1-\frac{1}{2^{i-2}}\right)\OPT, \text{ by Lemma~\ref{lemma:DeltaUpper}} \\
&=\left(2-\frac{1}{2^{i-2}}\right)\OPT.
\end{align*}

Likewise, at time $t'$, there are
\[ g_i(t') = a_i(t')+\Delta_i(t') \le \left(2-\frac{1}{2^{i-1}}\right)\OPT\]
packets alive that need to be processed on router $i$. 

Thus, if $p$ is released at router $i$, its completion time is no later than $t'+g_i(t')$.
If $p$ is released at router $i-1$, it reaches router $i$ no later than at time $t'+g_{i-1}(t')$.
Because no packet released at time $t'$ or later can delay $p$, it cannot be delayed on router $i$ beyond time $t'+g_i(t')$.
Therefore, the completion time of $p$ is at most $t'+\max\{g_{i-1}(t')+1,g_i(t')\}$.

Since $t'=r(p)+2$, we conclude that $p$ is completed no later than at time 
\begin{align*}
    r(p)+2+\max\{ g_{i-1}(t')+1, g_i(t') \} 
    &< r(p)+2 + \left(2-\frac{1}{2^{i-1}}\right)\OPT+1\\
    &\le r(p) + \left(2-\frac{1}{2^{k-1}}\right)\OPT+3, \text{ since } i \le k.
\end{align*}
 Hence, the flow time of $p$ is at most $\left(2-\frac{1}{2^{k-1}}\right)\OPT+3$.
\end{proof}

Since this matches the lower bound of Theorem~\ref{GreedyLower}, this proves that the competitive ratio of \GREEDY is exactly 
$\left(2-\frac{1}{2^{k-1}}\right)$ on inputs with only $k$ active routers. Since one is dealing with a restricted \OPT when using the online-bounded ratio, $\left(2-\frac{1}{2^{k-1}}\right)$ is also the online-bounded ratio.

\section{General lower bounds}\label{sec:generallower}

We start with a warm-up construction before proving the $4/3$ lower bound for randomized algorithms.

\begin{theorem}\label{thm:general-lower-rand}
	Any randomized algorithm has a competitive ratio of at least $6/5$ even with only $k=2$ routers.
\end{theorem}
\begin{proof}
    We give a family of input sequences $\{I_h\}_{h\in\NAT}$.
Each sequence $I_h$ begins with the following packets.
	\begin{itemize}
		\item At time $0$, $2\nrpackets$ packets are released that need to be processed on router $1$. We call these \emph{short packets}.
		\item At time $\nrpackets$, $\nrpackets$ packets are released that need to be processed on routers $1$ and $2$. We call these \emph{long packets}.
	\end{itemize}
	Let $y\in[0,\nrpackets]$ be the expected number of long packets that \alg processes on router $1$ before time $2h$.
    The expected maximum flow time among the short packets must be at least 
$2\nrpackets+y$.
If no further packets are released, the optimal maximum flow time is $2\nrpackets+1$, achieved by giving priority to the short packets at router $1$. 
Thus, since \alg is $6/5$-competitive, there must be a constant $a$ such that
\[2\nrpackets+y \leq \frac65 \cdot (2\nrpackets+1)+a,\]
or equivalently, by solving for $y$, 
\[
y\leq \frac{2\nrpackets}{5}+b, \text{ where } b=a+\frac65.
\]

    Thus, the expected number of packets that remain for \alg to process on router $2$ at time $2\nrpackets+1$ is at least 
    \[
    \nrpackets-y\geq\frac{3\nrpackets}{5}-b.
    \]
    If now $3\nrpackets$ additional \emph{jam packets} are released at time $2\nrpackets+1$ that need to be processed by router $2$, 
    the expected number of packets that \ALG still needs to process on router $2$ is at least 
    \[
    \frac{3\nrpackets}{5}-b+3h = \frac{18h}{5}-b
    \]
    so the expected maximum flow time in \ALG's schedule is at least this number.
    
    The optimal maximum flow time for the entire instance is $3\nrpackets$, achieved by giving priority to the long packets at router $1$ and thereby avoiding any conflict between the long packets and the jam packets. 
    Hence, $\ALG(I_h) \geq \frac65 \OPT(I_h)-b$. For $h$ large enough, this contradicts that \ALG is $(6/5-\eps)$-competitive.
\end{proof}

Next, we prove a lower bound of $4/3$ for randomized algorithms, using the same idea as in the proof of Theorem~\ref{thm:general-lower-rand}, using ``short'' and ``long'' packets iteratively to build up how much behind \OPT the randomized algorithm, \ALG, is on a specific router. The proof also reuses the idea of restricting
\ALG to maintaining the goal competitive ratio during each iteration (since it is online), but allowing \OPT to perform worse than \ALG until the end where there are many ``jam'' packets. Unfortunately, this means that the lower bound we obtain for the competitive ratio does not immediately imply the same result for the online-bounded ratio, even though the lower bound of Theorem~\ref{GreedyLower} for \GREEDY immediately implies a lower bound for \GREEDY's online-bounded ratio.

We call an instance $I$ $(t,i,U,L)$-{\em critical} if
there exists an offline algorithm \OFF satisfying the following conditions:
\begin{enumerate}
        \item[(i)] The cost of \OFF on $I$ is at most $U$.

        \item[(ii)] At time $t$, \OFF has no packets left to process on routers $i,i+1,\dots, k$.

        \item[(iii)] At time $t$, for every $4/3$-competitive randomized algorithm $\alg$,
        the expected number of packets that $\alg$ has left to process on router $i$ is at least $L$.
\end{enumerate}

First, we show the following lemma.

\begin{lemma}\label{lem:general-lower}
For any $4/3$-competitive algorithm \ALG, there exists a constant $b$ such that, for any nonnegative $L$ and any positive even $U$,
the existence of a $(t,i,U,L)$-critical instance, $I$, implies the existence of a $(t',i+1,3U/2,L+U/6-b)$-critical instance, $I'$, for some $t'\in\NAT$.
\end{lemma}

\begin{proof}
    We extend $I$ to $I'$ by releasing the following additional packets:
    \begin{itemize}
        \item At time $t$, $U$ \emph{short packets} that need to be processed at router $i$.
        \item At time $t+U/2$, $U/2$ \emph{long packets} that need to be processed at routers $i$ and $i+1$.
    \end{itemize}

    Let $P$ consist of the $L$ packets waiting on router $i$ at time $t$ and the $U$ packets released there at time $t$.     Let \alg be a $4/3$-competitive algorithm.
    Let $y\in[0,U/2]$ be the expected number of long packets released at time $t+U/2$ that \ALG processes on router $i$ before time $t+U$. Then the expected time by which \ALG has processed all packets in $P$ is at least $t+L+U+y$.
    Hence, the expected maximum flow time among the $L+U$ packets in $P$ is at least $L+U+y$.

    Based on \OFF, we can define another offline algorithm $\OFF'$ with $\OFF'(I')\leq U{+1}$ by processing all short packets before the long packets. 
    Because \ALG is $4/3$-competitive and 
    \OPT performs at least as well as $\OFF'$, there must exist a constant $a$, independent of $U$ and $L$, such that
    \[
    L+U+y \le \frac43(U+1)+a.
    \]
    Equivalently, 
    \[
    y\leq \frac U3-L+b, \text{ where } b=a+\frac43.
    \]
    Thus, the expected number of packets that \alg still has to process on router $i'=i+1$ at time $t'=t+U{+1}$ is at least
    \[
    \frac U2-y\geq \frac U6+L-b.
    \]
    Finally, from \OFF we construct another offline algorithm $\OFF''$ by processing the long packets as soon as they become available and before the remaining short packets. The schedule of $\OFF''$ has cost at most $(3/2)\cdot U$ and has no packets left to process on routers $i',i'+1,\dots,k$ from time $t'$.
\end{proof}

We now prove the theorem by iteratively applying this lemma.

\begin{theorem}\label{thm:general-lower}
    Any randomized algorithm has a competitive ratio of at least $4/3$.
\end{theorem}

\begin{proof}
    Let $\eps>0$. Suppose $\alg$ is 
    $(4/3-\eps)$-competitive.
    We start with the empty instance $I^{(0)}$, which is $(0,0,q,0)$-critical for any positive integer $q$. Let $i$ and $\ell$ be positive integers, where $\ell$ depends on $\varepsilon$.
    We set $q={2}^i\cdot \ell$, allowing us to apply Lemma~\ref{lem:general-lower} iteratively $i$ times to $I^{(0)}$ (the lemma requires the third entry of the $4$-tuple to be a multiple of {$2$}).

    This yields an instance $I^{(i)}$ that is $(t,i,U_i, L_i)$-critical for some $t\in\NAT$, where 
    \[
    U_i=\left(\frac32\right)^i \cdot q,
    \]
    and, for some constant $b$, 
    \begin{align*}
    L_i &= \sum_{j=0}^{i-1} \left( \frac16 \cdot U_j {-b} \right)
    {\geq} \left( \sum_{j=0}^{i-1} \left( \frac23 \right)^{i-j}  \cdot \frac16  \right)\cdot U_i {-bi}
    =\left(\sum_{j=1}^{i}\left(\frac23\right)^j\cdot\frac16\right)\cdot U_i{-bi}\\
    &= \left(\frac13-\frac{2^i}{3^{i+1}}\right)U_i-bi.
    \end{align*}
By choosing $i$ and $\ell$ sufficiently large, we obtain
\[
L_i > \left(\frac13-\frac{\eps}{4} \right) U_i - bi \geq \left( \frac13 - \frac{\eps}{2} \right) U_i.
\]
    Thus, by Lemma \ref{lem:general-lower}, \alg, being a $4/3$-competitive algorithm, has in expectation more than $(1/3 - \varepsilon/2)\cdot U_i$ packets left to process on router $i$ at time $t$.
    Releasing $U_i$ additional packets at time $t$ that need to be processed on router $i$ therefore gives \alg an expected cost greater than $(4/3 - \varepsilon/2)\cdot U_i$. In contrast,
    the definition of critical guarantees that \OFF, and hence \OPT, can schedule $I^{(i)}$ together with these additional packets with cost $U_i$, because no packet remains to be processed on router $i$ at time $t$. Since $\ell$, and hence $U_i$, can be chosen arbitrarily large, this contradicts the assumption that \alg is $(4/3-\eps)$-competitive.
\end{proof}

\section{Open Problem}

We conjecture that \GREEDY has a constant competitive ratio, possibly $2$, for online packet routing even when packet lengths are unbounded.
A difficulty in extending Lemma~\ref{lem:di+} to packets of length greater than two is that the packets delayed by \OPT, namely those in $G\setminus O$, need not all be located on one router. In the worst case, they could be distributed evenly across the preceding routers, for example. The maximum packet length may therefore enter the competitive analysis, perhaps only through an additive constant.

\bibliography{refs}

@article{DBLP:journals/algorithmica/AdamyAAE07,
  author       = {Udo Adamy and
                  Christoph Amb{\"{u}}hl and
                  R. Sai Anand and
                  Thomas Erlebach},
  title        = {Call Control in Rings},
  journal      = {Algorithmica},
  volume       = {47},
  number       = {3},
  pages        = {217--238},
  year         = {2007},
  doi          = {10.1007/s00453-006-0187-4}
}

@article{aiello2000adaptive,
  author    = {Aiello, W. and Kushilevitz, E. and Ostrovsky, R. and Ros{\'e}n, A.},
  title     = {Adaptive packet routing for bursty adversarial traffic},
  journal   = {Journal of Computer and System Sciences},
  volume    = {60},
  number    = {3},
  pages     = {482--509},
  year      = {2000},
  doi       = {10.1006/jcss.1999.1681}
}

@INPROCEEDINGS{ABCFMNP14,
   AUTHOR = "Antoniadis, Antonios and Barcelo, Neal and Cole, Daniel and Fox, Kyle and Moseley, Benjamin and Nugent, Michael and Pruhs, Kirk",
   TITLE = "Packet Forwarding Algorithms in a Line Network",
   BOOKTITLE = "11th Latin American Theoretical Informatics Symposium (LATIN)",
   SERIES = "Lecture Notes in Computer Science",
   VOLUME = 8392,
   PAGES = "610--621",
   PUBLISHER = "Springer",
   YEAR = 2014,
   doi = {10.1007/978-3-642-54423-1_53}
}

@inproceedings{awerbuch1993throughput,
  author    = {Awerbuch, Baruch and Azar, Yossi and Plotkin, Serge A.},
  title     = {Throughput-competitive on-line routing},
  booktitle = {34th Annual Symposium on Foundations of Computer Science (FOCS)},
  pages     = {32--40},
  year      = {1993},
  publisher = "IEEE",
  doi       = {10.1109/SFCS.1993.366884}
}

@article{DBLP:journals/mor/BansalKS06,
  author       = {Nikhil Bansal and
                  Tracy Kimbrel and
                  Maxim Sviridenko},
  title        = {Job Shop Scheduling with Unit Processing Times},
  journal      = {Mathematics of Operations Research},
  volume       = {31},
  number       = {2},
  pages        = {381--389},
  year         = {2006},
  doi          = {10.1287/moor.1060.0189}
}

@inproceedings{bohm2021throughput,
  title={Throughput scheduling with equal additive laxity},
  author={B{\"o}hm, Martin and Megow, Nicole and Schl{\"o}ter, Jens},
  booktitle={12th International Conference on Algorithms and Complexity (CIAC)},
  pages={130--143},
  year={2021},
  series       = {Lecture Notes in Computer Science},
  volume       = {12701},
  publisher    = {Springer},
  doi          = {10.1007/978-3-030-75242-2_9}
}

@article{borodin2001adversarial,
  author    = {Borodin, Allan and Kleinberg, Jon M. and Raghavan, Prabhakar and Sudan, Madhu and Williamson, David P.},
  title     = {Adversarial queuing theory},
  journal   = {Journal of the ACM},
  volume    = {48},
  number    = {1},
  pages     = {13--38},
  year      = {2001},
  doi       = {10.1145/363647.363659}
}

@article{BEFLL18,
  title={Online-bounded analysis},
  author={Boyar, Joan and Epstein, Leah and Favrholdt, Lene M. and Larsen, Kim S. and Levin, Asaf},
  journal={Journal of Scheduling},
  volume={21},
  number={4},
  pages={328--441},
  year={2018},
  doi={10.1007/s10951-017-0536-y}
}

@article{Chen01021995,
author = {Bo Chen},
title = {Analysis of Classes of Heuristics for Scheduling a Two-Stage Flow Shop with Parallel Machines at One Stage},
journal = {Journal of the Operational Research Society},
volume = {46},
number = {2},
pages = {234--244},
year = {1995},
doi = {10.1057/jors.1995.28}
}

@article{chlebus2012adversarial,
  author    = {Chlebus, Bogdan S. and Kowalski, Dariusz R. and Rokicki, Mariusz A.},
  title     = {Adversarial queuing on the multiple access channel},
  journal   = {ACM Transactions on Algorithms},
  volume    = {8},
  number    = {1},
  pages     = {5},
  year      = {2012},
  doi       = {10.1145/2071379.2071384}
}

@inproceedings{das2025approximation,
  title={Approximation Hardness of Resource Scheduling},
  author={Das, Rathish and Sun, Hao},
  booktitle={37th ACM Symposium on Parallelism in Algorithms and Architectures (SPAA)},
  pages={46--61},
  year={2025},
  publisher = "ACM",
  doi       = {10.1145/3694906.3743340}
}

@inproceedings{DBLP:conf/icalp/DisserKL15,
  author       = {Yann Disser and
                  Max Klimm and
                  Elisabeth L{\"{u}}bbecke},
  _editor       = {Magn{\'{u}}s M. Halld{\'{o}}rsson and Kazuo Iwama and Naoki Kobayashi and Bettina Speckmann},
  title        = {Scheduling Bidirectional Traffic on a Path},
  booktitle    = {42nd International Colloquium on Automata, Languages, and Programming (ICALP)},
  series       = {Lecture Notes in Computer Science},
  volume       = {9134},
  pages        = {406--418},
  publisher    = {Springer},
  year         = {2015},
  doi          = {10.1007/978-3-662-47672-7_33}
}

@INPROCEEDINGS{667220,
  author={Erlebach, Thomas and Jansen, Klaus},
  booktitle={30th Annual Hawaii International Conference on System Sciences (HICSS)}, 
  title={Call scheduling in trees, rings and meshes}, 
  year={1997},
  publisher = "IEEE",
  volume={1},
  pages={221},
  doi={10.1109/HICSS.1997.667220}
}

@inproceedings{DBLP:conf/icalp/Havill01,
  author       = {Jessen T. Havill},
  _editor       = {Fernando Orejas and
                  Paul G. Spirakis and
                  Jan van Leeuwen},
  title        = {Online Packet Routing on Linear Arrays and Rings},
  booktitle    = {28th International Colloquium on Automata, Languages and Programming (ICALP)},
  series       = {Lecture Notes in Computer Science},
  volume       = {2076},
  pages        = {773--784},
  publisher    = {Springer},
  year         = {2001},
  doi          = {10.1007/3-540-48224-5_63}
}

@article{hyatt2024approximations,
  title={Approximations for throughput maximization},
  author={Hyatt-Denesik, Dylan and Rahgoshay, Mirmahdi and Salavatipour, Mohammad R},
  journal={Algorithmica},
  volume={86},
  number={5},
  pages={1545--1577},
  year={2024},
  publisher={Springer},
  doi={10.1007/s00453-023-01201-4}
}

@inproceedings{DBLP:conf/spaa/ImM15a,
  author       = {Sungjin Im and
                  Benjamin Moseley},
  editor       = {Guy E. Blelloch and
                  Kunal Agrawal},
  title        = {Scheduling in Bandwidth Constrained Tree Networks},
  booktitle    = {27th {ACM} on Symposium on Parallelism in Algorithms
                  and Architectures (SPAA)},
  pages        = {171--180},
  publisher    = {{ACM}},
  year         = {2015},
  doi          = {10.1145/2755573.2755576}
}

@article{kesselman2004buffer,
  author    = {Kesselman, Alexander and Lotker, Zvi and Mansour, Yishay and Patt-Shamir, Boaz and Schieber, Baruch and Sviridenko, Maxim},
  title     = {Buffer overflow management in {QoS} switches},
  journal   = {SIAM Journal on Computing},
  volume    = {33},
  number    = {3},
  pages     = {563--583},
  year      = {2004},
  doi       = {10.1137/S0097539701399666}
}

@article{leighton1994packet,
  author    = {Leighton, F. T. and Maggs, Bruce M. and Rao, Satish B.},
  title     = {Packet routing and job-shop scheduling in \mbox{$O$}(Congestion + Dilation) steps},
  journal   = {Combinatorica},
  volume    = {14},
  number    = {2},
  pages     = {167--186},
  year      = {1994},
  doi       = {10.1007/BF01215349}
}

@article{liu2024online,
  title={Online task scheduling and termination with throughput constraint},
  author={Liu, Qingsong and Fang, Zhixuan},
  journal={IEEE/ACM Transactions on Networking},
  volume={32},
  number={6},
  pages={4629--4643},
  year={2024},
  publisher={IEEE},
  doi       = {10.1109/TNET.2024.3425617}
}

@inproceedings{ostrovsky1997universal,
  author    = {Ostrovsky, Rafail and Rabani, Yuval},
  title     = {Universal \mbox{$O$}(Congestion + Dilation + $\mbox{log}^{1+\epsilon} \mbox{$N$}$) local control packet switching algorithms},
  booktitle = {29th Annual ACM Symposium on Theory of Computing (STOC)},
  pages     = {644--653},
  year      = {1997},
  publisher = "ACM",
  doi       = {10.1145/258533.258659}
}

@inproceedings{DBLP:conf/infocom/Patt-ShamirR19,
  author       = {Boaz Patt{-}Shamir and
                  Will Rosenbaum},
  title        = {Space-Optimal Packet Routing on Trees},
  booktitle    = {{IEEE} Conference on Computer Communications (INFOCOM)},
  pages        = {1036--1044},
  publisher    = {{IEEE}},
  year         = {2019},
  doi          = {10.1109/INFOCOM.2019.8737596}
}

@article{SOUKHAL200532,
author = {Ameur Soukhal and Ammar Oulamara and Patrick Martineau},
title = {Complexity of flow shop scheduling problems with transportation constraints},
journal = {European Journal of Operational Research},
volume = {161},
number = {1},
pages = {32--41},
year = {2005},
doi = {10.1016/j.ejor.2003.03.002}
}

@article{doi:10.1137/1.9781611975482.9,
author = {Pavel Veselý and Marek Chrobak and Łukasz Jeż and Jiří Sgall},
title = {A $\phi$-Competitive Algorithm for Scheduling Packets with Deadlines},
journal = "SIAM Journal on Computing",
volume = 51,
number = 5,
year = 2022,
pages = {1626-1691},
doi = {10.1137/21M1469753}
}

@article{WEI201941,
author = {Hongjun Wei and Jinjiang Yuan},
title = {Two-machine flow-shop scheduling with equal processing time on the second machine for minimizing total weighted completion time},
journal = {Operations Research Letters},
volume = {47},
number = {1},
pages = {41--46},
year = {2019},
doi = {10.1016/j.orl.2018.12.002}
}
\bibliographystyle{plain}

\end{document}